\documentclass[12pt]{article}

\usepackage{amsfonts,amsmath,amsthm}
\usepackage{graphicx}
\newtheorem{theorem}{Theorem}[section]

\theoremstyle{definition}
\newtheorem{definition}[theorem]{Definition}
\newtheorem{definitions}[theorem]{Definitions}
\newtheorem{example}{Example}
\usepackage{algorithm}%
\usepackage{algorithmicx}%
\usepackage{algpseudocode}%
\usepackage{listings}%
\lstdefinelanguage{Maxima}{
keywords={linsolve,append,length,col,transpose,matrix,addrow,rhs,
submatrix,in,cons,or,else,elseif,askinteger,makelist,return,block,zeromatrix,read,for,thru,do,if,else,then,print},
sensitive=true,comment=[n][\itshape]{/*}{*/}}
\title{On the minimal simplex economy}
\author{Antonio Pulgar\'{i}n}
\date{\small{Univ. of Extremadura, Avda. Universidad s/n, 10003 C\'{a}ceres (Spain)\\
email: aapulgar@unex.es,\ ORCID: 0000-0001-5367-2888}}
\begin{document}
\maketitle 
\begin{abstract}In our previous paper we proved that every affine economy has a competitive equilibrium. We define a simplex economy as an affine economy consisting of a stochastic allocation (defining the initial endowments) and a variation with repetition of the number of commodities taking the number of consumers (representing the preferences). We show that a competitive equilibrium can be intrinsically computed in any minimal simplex economy.\par\medskip\noindent
Keywords: Simplex economy, competitive equilibrium, stochastic allocation.\par\medskip\noindent
JEL Classification: {D50}\par\noindent
MSC Classification: {91B50}
\end{abstract}
\section{Introduction}\label{sec1}
The commodity space should reflect the constraints limiting consumption possibilities. For example, a lower bound 0 would imply that one cannot consume a negative quantity of a good, and an upper bound 1 would mean that agents cannot consume more than entirety the good. This fact suggest assuming commodities as proportions of the closed interval $[0,1]$.
\begin{definition}
The {\itshape commodity space} for a finite number $n$ of consumption goods is the product $$[0,1]^n$$where each component represents the proportion of the corresponding commodity among all possible combinations.\end{definition} 

The value of a commodity will depend on the other commodities that can be obtained in exchange for it, and to this aim we need to consider how the agents allocate their income among the different commodities.
\begin{definitions}The simplex $$P=\left\lbrace(p_1,\ldots,p_n)\in[0,1]^n:p_1+\dots+p_n=1\right\rbrace,$$ endowed with the subspace topology induced from the canonical product topology of $[0,1]^n$ is called {\itshape price space}.

The {\itshape extreme points boundary} of $P$ is the finite subset $$\partial P=\{e_1,\ldots,e_n\}\subset P,\text{ where } e_j=(\delta_{1j},\ldots,\delta_{nj})\text{ with }\delta_{ij}=\begin{cases}0\text{ if }i\neq j\\
1\text{ if }i=j
\end{cases}$$\end{definitions} 

The {\itshape value} of the {\itshape commodity bundle} $f=(a_1,\ldots,a_n)\in [0,1]^n$
for the {\itshape price system} $p=(p_1,\ldots,p_n)\in P$ is
$$f\cdot p=a_1p_1+\cdots+a_np_n\in[0,1],$$
thus, 
$$a_j=f\cdot e_j\text{ for all }j=1,\ldots ,n.$$
The commodity space $[0,1]^n$ is therefore isomorphic to the vector lattice $C(\partial P,[0,1])$ of $[0,1]$-valued continuous functions on the finite subset $\partial P$:$$(a_1,\ldots,a_n)\in [0,1]^n\longleftrightarrow f\in C(\partial P,[0,1]):f(e_j)=f\cdot e_j=a_j\in[0,1].$$
\begin{figure}[h]
\centering
\scalebox{0.4}{\includegraphics{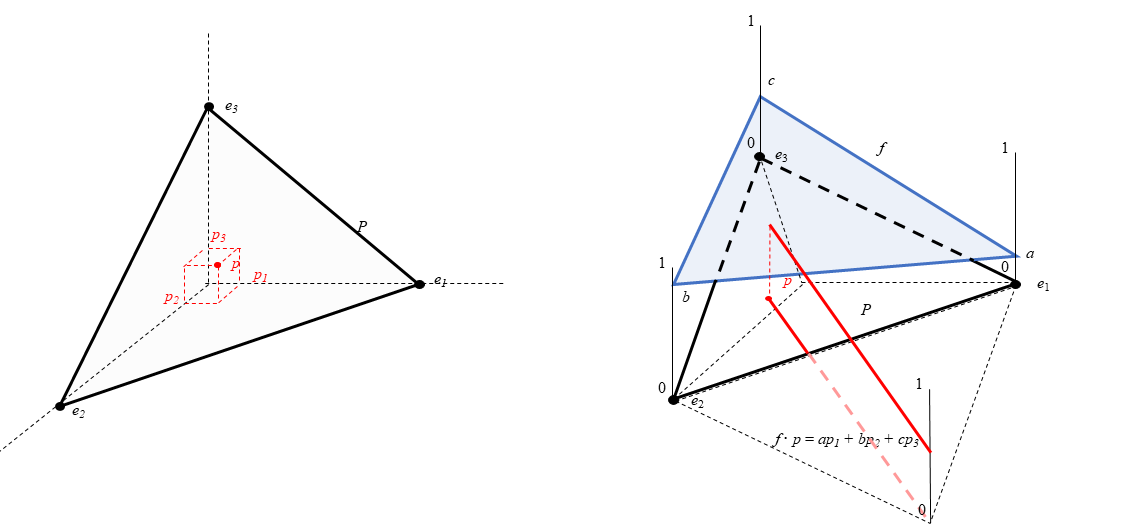}}
\caption{$3$-commodity-price space and the representation of a commodity bundle.}\label{figure1}
\end{figure}
\section{The simplex economy}\label{sec2}
Let $n$ be the number of commodities and consider a finite number $m$ of consumer agents. 
\begin{definitions}\label{stochastic}
Matrices $$F=\left(\begin{matrix}
a_{11}&\dots & a_{1n}\\
\vdots &\ddots& \vdots\\
a_{m1}&\dots & a_{mn}
\end{matrix}\right)$$ whose elements are probabilities (i.e. $a_{ij}\in[0,1]$) are called 
{\itshape allocations}. Their rows $$f_i=(a_{i1},\ldots, a_{in})\in [0,1]^n$$identify with the family of commodity bundles $\{f_1,\ldots,f_m\}$ where $$a_{ij}=f_i\cdot e_j.$$

If $a_{1j}+\ldots +a_{mj}=1$ for all $j=1,\ldots,n$, then the matrix $F$ is said to be a {\itshape stochastic allocation}.

Each consumer $i=1,\ldots,m$ will contribute to the market with an {\itshape initial endowment} $\omega_i=(w_{i1},\ldots, w_{in})$ defined from the row $i$ of a stochastic allocation
$$W=\left(\begin{matrix}
w_{11}&\dots & w_{1n}\\
\vdots &\ddots& \vdots\\
w_{m1}&\dots & w_{mn}
\end{matrix}\right).$$

The family of initial endowments $\{\omega_1,\ldots,\omega_m\}$ becomes a partition of the unity, that is, the {\itshape total endowment} $$\omega=\omega_1+\cdots+\omega_m$$ satisfies
$\omega\cdot p=1$ for all $p\in P$. 

If an allocation $F$ clears the market (i.e. 
$(f_1+\cdots+f_m)\cdot p=\omega\cdot p$ for all $p\in P$), then $F$ is said to be {\itshape feasible}.

Each consumer $i=1,\ldots,m$ chooses commodity bundles accordingly to a {\itshape preference} relationship
represented by a utility function $$u_i(f)=f\cdot e_{\sigma(i)},\text{ where
}\sigma(i)\in \{1,\ldots,n\}.$$

The {\itshape strict preference} relationship between allocations is defined as follows:
$$F\prec G\iff  f_{i}\cdot e_{\sigma(i)} <g_{i}\cdot e_{\sigma(i)}\ \text{ for all }i=1,\ldots,m,$$where $\{f_1,\ldots,f_m\}$ and $\{g_1,\ldots,g_m\}$ are the families of commodity bundles defined by $F$ and $G$  respectively.\end{definitions}
\begin{definition}\label{simplex}A {\itshape simplex economy} is a dupla $\langle W,\sigma\rangle$ where 
$$W=\left(\begin{matrix}
w_{11}&\dots & w_{1n}\\
\vdots &\ddots& \vdots\\
w_{m1}&\dots & w_{mn}
\end{matrix}\right)$$is a stochastic allocation defining the initial endowments, and 
$$\sigma=(\sigma(1),\ldots,\sigma(m))$$ is a variation with repetition of the $n$ commodities representing the preferences of the $m$ consumer agents.
\end{definition}
\section{Competitive equilibrium}\label{sec3}
\begin{definition}Given a simplex economy  $\langle W,\sigma\rangle$, a feasible allocation $F$ is said to be a {\itshape competitive equilibrium} provided there exists a price system $p\in P$, $p\neq 0$ for which $F\prec G$ implies $G$ is out of the $p$-budget (i.e. there exists $i\in\{1,\ldots,m\}$ such that $g_{i}\cdot p> w_{i}\cdot p)$.
\end{definition}
\begin{theorem}Every simplex economy has a competitive equilibrium
\end{theorem}
\begin{proof}
Given $p=(p_1,\ldots,p_n),q=(q_1,\ldots,q_n)\in P$, $0\leq \lambda\leq 1$ and $f=(a_1,\ldots,a_n)\in[0,1]^n=C(\partial P,[0,1])$:\begin{align*}
f\cdot (\lambda p+(1-\lambda)q)=&\,a_1(\lambda p_1+(1-\lambda)q_1)+\dots+a_n(\lambda p_n+(1-\lambda)q_n)\\
=&\,\lambda(a_1p_1+\dots +a_1p_n)+(1-\lambda)(a_1q_1+\dots +a_nq_n)\\
=&\,\lambda (f\cdot p)+(1-\lambda)(f\cdot q)
\end{align*} and then, $$f:p\in P\mapsto f\cdot p\in[0,1]$$ becomes an affine continuous function on $P$. Accordingly, $P$ is a Bauer simplex.

It is straightforward to see that every simplex economy becomes an affine economy and we conclude by applying  \cite[Theorem 4.4]{MR4573494}.
\end{proof}
While the existence of competitive equilibrium in abstract economic models holds theoretical importance, its practical relevance hinges on our ability to compute such equilibria in concrete settings. This study addresses this gap by presenting a novel approach to directly compute equilibria in the context of the simplex economies.

For any simplex economy $\langle W,\sigma\rangle$ denote $$\begin{array}{rcl}
\sigma(i^{1}_1)=&\dots &=\sigma(i^{1}_{m_1})=j_1\\
\vdots&\ &\vdots\\
\sigma(i^{k}_1)=&\dots&=\sigma(i^{k}_{m_k})=j_k
\end{array}$$ where $i_r^s=1,\ldots,m$, $j_1\neq\dots\neq j_k\in\{1,\ldots,n\}$ and $m_1+\cdots+m_k=m$.
\subsection{The feasible allocation}
\begin{definition}\label{F*}We denote by $F^*$ the {\itshape allocation}
 consisting of the family $\{f^*_1,\ldots,f^*_m\}$
of commodity bundles defined by $$f^*_{i_r^s}\cdot e_j=\begin{cases}w_{{i_r^s}{j_s}}+\dfrac{(m-m_s)\min\left\{w_{i{j_s}}:i\neq i_{1}^s,\ldots,i_{m_s}^s\right\}}{m_s}\  \text{ if }j=j_s;\\
w_{{i_r^s}{j_t}}-\min\left\{w_{i{j_t}}:i\neq i_{1}^t,\ldots,i_{m_t}^t\right\}\  \text{ if }j=j_t,\ t\neq s;\\
\dfrac{1}{m}\ \text{ if }j\neq j_1,\ldots,j_k.\end{cases}$$\end{definition}  
For every $j=1,\ldots,n$
$$\left(f^*_{1}+\dots+ f^*_{m}\right)\cdot e_{j}=1.$$
Furthermore, 
\begin{align*}(f^*_1+\cdots+f^*_m)\cdot p&=f^*_1\cdot p+\dots+f^*_m\cdot p\\
&=\sum_{j=1}^{n}\left(f^*_1\cdot e_j\right)p_j+\dots+\sum_{j=1}^{n}\left(f^*_m\cdot e_j\right)p_j\\&=
\sum_{j=1}^{n}\left(\left(f^*_1+\dots+f^*_m\right)\cdot e_j\right)p_j
\\&=
p_1+\cdots+p_n=1=\omega\cdot p\text{ for all }p\in P, \end{align*}which implies that $F^*$ becomes a feasible allocation. 
\subsection{The supporting price}\label{sp}
On constructing $F^*$ we are ensuring that the following linear system
\begin{equation}\label{supp}\left(\begin{array}{ccc}
\left(f^*_{1}-\omega_{1}\right)\cdot e_{j_1}&  \dots &\left(f^*_{1}-\omega_{1}\right)\cdot e_{j_k}\\
\vdots & \ddots& \vdots\\
\left(f^*_{m}-\omega_{m}\right)\cdot e_{j_1}&  \dots &\left(f^*_{m}-\omega_{m}\right)\cdot e_{j_k}\\
1&\dots&1
\end{array}\right)
\left(\begin{matrix}
p^*_{j_1}\\
\vdots\\
p^*_{j_k}
\end{matrix}\right)= \left(\begin{matrix}
0\\
\vdots\\
0\\
1
\end{matrix}\right)
\end{equation}
is compatible and determined.
\begin{definition}\label{p*}The {\itshape supporting price} for the feasible allocation $F^*$ is the price system 
$p^*=(p^*_1,\ldots,p^*_n)\in P$, where 
$p^*_{j_1},\ldots,p^*_{j_k}$ is the solution of the above system (\ref{supp}) and $p^*_j=0$ for $j\neq j_1,\ldots,j_k$.\end{definition}Supporting price $p^*$ satisfies:
 $$f^*_{i}\cdot p^*=\omega_{i}\cdot p^*\text{ for all } i=1,\ldots,m,$$for the feasible allocation $F^*$.
\begin{figure}[h]
\centering
\scalebox{0.4}{\includegraphics{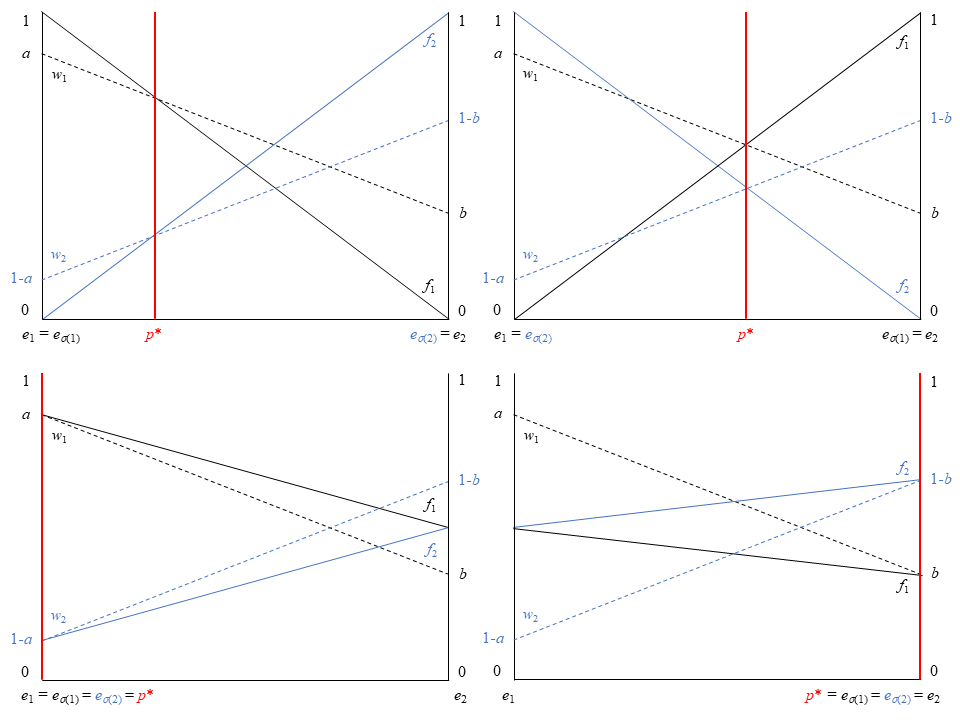}}
\caption{2-commodity-price space and the representation of 2-consumers competitive equilibrium}\centering\label{figure2}
\end{figure}
\begin{definition}\label{minimal}A simplex economy $\langle W,\sigma\rangle$ is said to be {\itshape minimal} provided there exists $i_r^s\in\{1,\ldots,m\}$ such that $$w_{{i_r^s}{j_t}}=\min\left\{w_{i{j_t}}:i\neq i_{1}^t,\ldots,i_{m_t}^t\right\}\text{ for all } t\neq s.$$ 
\end{definition}
\begin{theorem}\label{main}
If $\langle W,\sigma\rangle$ is a minimal simplex economy, then the feasible allocation $F^*$ is a competitive equilibrium supported by the price system $p^*$.
\end{theorem} 
\begin{proof}
On the one hand, minimality condition implies that there exists $i_r^s\in\{1,\ldots,m\}$ such that $$f^*_{i_r^s}\cdot e_{j_t}=0 \text{ for all }t\neq s.$$
Since $p^*_j=0$ for $j\neq j_1,\ldots,j_k$ then
$$f^*_{i_r^s}\cdot p^*=\sum_{j=1}^{n}\left(f^*_{i_r^s}\cdot e_{j}\right)p^*_{j}=\left(f^*_{i_r^s}\cdot e_{j_s}\right)p^*_{j_s}.$$  

On the other hand, if $F^*\prec G$ and $G$ is
within the $p^*$-budget, then 
$$f^*_{i_r^s}\cdot e_{j_s}<g_{i_r^s}\cdot e_{j_s}\text{ and }g_{i_r^s}\cdot p^*\leq \omega_{i_r^s}\cdot p^*$$respectively. Therefore
\begin{align*}
f^*_{i_r^s}\cdot p^*=\left(f^*_{i_r^s}\cdot e_{j_s}\right)p^*_{j_s}<\left(g_{i_r^s}\cdot e_{j_s}\right)p^*_{j_s}\leq g_{i_r^s}\cdot p^*\leq 
\omega_{i_r^s}\cdot p^*=f^*_{i_r^s}\cdot p^*,
\end{align*}which is a contradiction. Hence, $F^*$ is a competitive equilibrium supported by $p^*$.
\end{proof}
\subsection{Open question}
We do not know whether $F^*$ remains a competitive equilibrium supported by $p^*$ if we unassume minimality condition (recall that if minimality does not hold, then $f^*_{i}\cdot p^*>  \left(f^*_{i}\cdot e_{j}\right)p^*_{j}$ for all $i,j$).
\section{Computing an example: Maxima source code}
 Theorem \ref{main} establishes the existence  of a competitive equilibrium in any minimal simplex economy. 

By means of the computer algebra system \cite{maxima} we develop a source code computing such a equilibrium. We illustrate the algorithm achievements throughout  a numerical example:
\begin{algorithm}[h]
\caption{Inputs and checking simplexity}\label{a1}
\begin{lstlisting}[language=Maxima]
block(
  n:read("Number of commodities:"),  
  m:read("Number of consumers:"),           
  W:zeromatrix(m,n),
  d:makelist(0,j,1,n),
  s:makelist(0,i,1,m),
  N:makelist(j,j,1,n), 
  for j:1 thru n do(   
    for i:1 thru m do(
      W[i,j]:read("w(",i,",",j,"):"),
      if(W[i,j]<0 or W[i,j]>1) then(
        print("not in [0,1]"),W[i,j]:0,i:i-1)
      else(d[j]:d[j]+W[i,j]))),
  for i:1 thru m do(
    s[i]:read("sigma(",i,"):"),
    if(askinteger(s[i])=no) then(
      print("not an integer"),i:i-1) 
    elseif(s[i]<1 or s[i]>n) then(
      print("not in ",N),i:i-1)),
  return("W"=W),
  return("sigma"=s),
  return("Sum"=d))
\end{lstlisting}
\end{algorithm}
\begin{example}\label{exa}Consider the minimal simplex economy consisting of $n=4$ commodities and $m=5$ consumer agents, having as initial endowments the stochastic allocation   
$$W=\left(\begin{matrix}
\mathbf{0.2}& 0.4 & \mathbf{0.1}&\mathbf{0.1}\\
{0.2} &0.3& 0.2 &0.4\\
0.2&0.2 & 0.2&0.3\\
0.2&0.1&0.3&{0.1}\\
0.2&0&0.2&{0.1}
\end{matrix}\right),$$
and preferences represented by the variation with repetition $$\sigma=(1,1,3,4,4).$$ 
\end{example}
Algorithm \ref{a1} asks for $W$ and $\sigma$ and checks if they satisfy required conditions of Definition \ref{simplex}. We have not included checking stochastic condition since initial endowments could to be provided with no exact decimals, and therefore the sum of the respective columns have not to exactly equal 1 unvalidating the stochastic condition. Nevertheless, algorith return vector $Sum$ to freely deciding whether the matrix holds to be stochastic or not.  
\begin{algorithm}[h]
\caption{Checking minimality}\label{a2}
\begin{lstlisting}[language=Maxima]
block(
  Min:makelist(1,j,1,n),
  t1:makelist(1,i,1,m),
  t0:makelist(0,i,1,m),
  for j:1 thru n do(
    for i:1 thru m do(
      if(s[i]#j) then( 
        if(Min[j]>W[i,j]) then(
          Min[j]:W[i,j])))),
  for i:1 thru m do(
    for j in s do(
      if(j#s[i] and W[i,j]#Min[j]) then(
        t1[i]:0))),
  if(t1=t0) then(
    print("The simplex economy is not minimal")) else(
    print("The simplex economy is minimal")),
  return("Min"=Min))
\end{lstlisting}
\end{algorithm}

Algorithm \ref{a2} checks Definition \ref{minimal} requirements. Notice that for $i_1^1=1$ it is satisfied \begin{align*}
&\mathbf{0.1}=w_{13}=\min\{w_{13},w_{23},w_{43},w_{53}\}=\min\{0.1,0.2,0.3,0.2\}\\
&\mathbf{0.1}=w_{14}=\min\{w_{14},w_{24},w_{34}\}=\min\{0.1,0.4,0.3\}
\end{align*}
and therefore $\langle W,\sigma\rangle$ is minimal. 

\begin{algorithm}[h]
\caption{Computing the feasible allocation $F^*$}\label{a3}
\begin{lstlisting}[language=Maxima]
block(
  F:zeromatrix(m,n),  
  M:makelist(0,j,1,n),
  for j in s do(
     M[j]:M[j]+1),
  for i:1 thru m do(
    for j:1 thru n do(
      F[i,j]:1/m),
    for j in s do(
      if(s[i]=j) then(
        F[i,j]:W[i,j]+((m-M[j])/M[j])*Min[j]) else(
          F[i,j]:W[i,j]-Min[j]))),
  return("F*"=F))
\end{lstlisting}
\end{algorithm}

Algorithm \ref{a3} computes the feasible allocation $F^*$ as in Definition \ref{F*} by obtaining $$F^*=\left(\begin{matrix}
{0.5}& 0.2& 0&0\\
{0.5} &0.2& 0.1&0.3\\
0&0.2 & {0.6}&0.2\\
0&0.2 & 0.2&{0.25}\\
0&0.2 &0.1&{0.25}
\end{matrix}\right)$$ 

Observe from $\sigma=(1,1,3,4,4)$ that consumers $1,2$ prefer commodity $1$. Both obtain $0.5$ after the feasible distribution while they contributed with $0.2$. Consumer 3 prefers commodity 3 and receives $0.6$ once it has apported $0.2$. Consumers $4,5$ prefer commodity 4 and obtain $0.25$ once contributing both with $0.1$. All consumers improve according with their preferences after the feasible allocation.
Commodity 2 is not preferred by any consumer agent. The feasible allocation suggests dividing the entire good equally among the consumers, by assigning the amount 0.2 to each one, regardless of their contributions.
\begin{algorithm}[h]
\caption{Computing the supporting price $p^*$}\label{a4}
\begin{lstlisting}[language=Maxima]
block(
  pvp:makelist(0,j,1,n),                 
  p:transpose(matrix(makelist('p[j],j,1,n))),
  S:addrow(F-W,makelist(1,j,1,n)),
  for j:1 thru n do(
    if(col(F,j)=transpose(matrix(makelist(1/m,i,1,m)))) then(
      S:submatrix(S,j),
      p:submatrix(j,p),
      pvp[j]:1)),
  B:addrow(transpose(makelist(0,i,1,m)),[1]),
  p:linsolve(transpose(S.p-B)[1],transpose(p)[1]),
  for j:1 thru n do(
    if(pvp[1]=1) then(
      p:append(cons(0,makelist(rhs(p[i]),i,1,length(p)))))
    elseif(j>1 and pvp[j]=1) then(        
      p:append(makelist(rhs(p[i]),i,1,j-1),cons(
        0,makelist(rhs(p[i]),i,j,length(p)))))),
  return("p*"=p))
\end{lstlisting}
\end{algorithm}

Algorithm \ref{a4} computes the supporting price $p^*$ as in Definition \ref{p*} by solving the linear system
\[\left(\begin{array}{rrr}
{0.3}&  -0.1 &-0.1\\
{0.3} & -0.1& -0.1\\
-0.2& {0.4}& -0.1\\
-0.2& -0.1&0.15\\
-0.2&-0.1&0.15\\
1&1&1
\end{array}\right)
\left(\begin{matrix}
p^*_1\\
p^*_3\\
p^*_4
\end{matrix}\right)= \left(\begin{matrix}
0\\
0\\
0\\
0\\
0\\
1
\end{matrix}\right)\]and obtaining $$p^*=\left(0.25,0,0.25,0.5\right).$$

Supporting price remains the same value of the bundles that the initial endowments.

Recall that commodity 4 has the highest price due to it is highly demanded, and a relatively small amount is provided by the interested agents.
Initial endowment $\omega_2$ has the highest value $\omega_2\cdot p^*=0.3$ and conversely $\omega_1$ has the lowest value $\omega_1\cdot p^*=0.125$ (consumer 2 contributed to the market with a big amount of commodity $4$ against the small amount of consumer 1).
\subsection{Computational complexity}
Our approach avoids the usage of numerical methods or approximation algorithms, which may be inaccurate, unstable, or inefficient. Instead, it uses symbolic computation and exact arithmetic, which guarantees the correctness and precision of the solution. Furthermore, the approach takes advantage of the minimal simplex economy structure, which simplifies the system of equations and reduces the number of unknowns. This makes the computation faster and easier, as demonstrated by the Maxima source code.

Computational cost of the algorithm can be approximated as $O(n*m)$, where $n$ is the number of commodities and $m$ is the number of consumers. It is worth noting that it includes several checks for input validity, although important for ensuring correct program behavior, they also add to the computational cost. However, since these checks are also within the loops, they do not change the overall time complexity.
\section{Concluding remarks}
\subsection*{Economic implications} 
While theoretical proof of the existence of a  competitive equilibrium is hugely valuable, our ability to compute it unlocks doors to understanding and influencing real-world markets.

For instance, by unveiling hidden efficiency of a seemingly well-functioning market. Authorities often grapple with simulate the impact of policies (taxes, subsidies, etc.). By computing the resulting equilibrium, they can predict changes in prices and  allocations producing a bigger individual welfare, leading to more  impactful policy choices. Supply chains, public resource management, and market design  require balancing efficiency with fairness. Modeling these systems  and being able to compute equilibria  can translate to lower costs, better resource utilization, and improved satisfaction for all stakeholders.

Moving forward, extending this approach to broader, more complex economic structures would be crucial to unlock its full potential and pave the way for a deeper understanding of how economies behave and evolve.
\subsection*{Real-time implementation}
Imagine an online auction platform where buyers and sellers interact in real time. As demand fluctuates, prices should adjust dynamically ensuring efficient transactions. 

While theoretical foundation provides a solid framework for understanding market dynamics and resource allocation, real-time changes require adaptive algorithms that respond promptly, thus ensuring accurate, timely data is crucial for reliable equilibrium computation. 

To this aim, we may set dynamic simplex economies as duplas $\langle W(t),\sigma(t)\rangle$ where 
$W(t)$ is a dynamic stochastic allocation defining real-time endowments $\{\omega_1(t),\ldots,\omega_m(t)\}$, as continuous functions 
$$\omega_i:\mathbb{R}_+\to [0,1]^n,\ t\mapsto \omega_i(t)=(w_{i1}(t),\ldots,w_{in}(t))\ \text{ for }i=1,\ldots,m,$$and  
$$\sigma:\mathbb{R}_+\to \{1,\ldots,n\}^m,\ t\mapsto \sigma(t)=(\sigma_1(t),\ldots,\sigma_m(t))$$representing the dynamic preferences of the consumer agents. 

Therefore, our algorithm could be easily adapted to real-time data by responding dynamically to changing economic conditions and compute prices based on bid and ask data.

\end{document}